\documentclass[a4paper, 12pt]{amsart}

\usepackage{amssymb, amsthm}
\usepackage{amsmath,amsfonts,amssymb,amscd,amsthm,amsbsy,upref}
\usepackage{makecell}
\usepackage{yhmath}
\usepackage{adjustbox}
\usepackage[backref]{hyperref}
\usepackage[alphabetic,backrefs,lite]{amsrefs}
\usepackage{amscd}   % for commutative diagrams
\usepackage[all]{xy} % for complicated commutative diagrams
\usepackage{graphicx}
\usepackage{mathrsfs}
\usepackage{enumerate}

\usepackage{tikz}
\usepackage{verbatim}

  \tikzstyle{block} = [rectangle, draw,
      text width=10em, text centered, minimum height=4em]
  \tikzstyle{line} = [draw, -latex']

\usepackage{amssymb, amsmath, amsthm}
\usepackage[backref]{hyperref}
\usepackage[alphabetic,backrefs,lite]{amsrefs}
\usepackage{amscd}   % for commutative diagrams
\usepackage[all]{xy} % for complicated commutative diagrams
\usepackage{graphicx}
\usepackage{mathrsfs}
\definecolor{labelkey}{rgb}{0,0,1}
\usepackage{}% http://ctan.org/pkg/multirow
\usepackage{tikz}

\newtheorem{lemma}{Lemma}[section]
\newtheorem{theorem}[lemma]{Theorem}
\newtheorem{proposition}[lemma]{Proposition}

\newtheorem*{prop*}{Proposition}

\newtheorem{claim*}{Claim}
\newtheorem{thm}[lemma]{Theorem}
\newtheorem{defn}[lemma]{Definition}

\theoremstyle{definition}

\newtheorem{corollary}[lemma]{Corollary}
\newtheorem{example}[lemma]{Example}

\makeatletter
\def\namedlabel#1#2{\begingroup
    #2%
    \def\@currentlabel{#2}%
    \phantomsection\label{#1}\endgroup
}
\makeatother

\newcommand{\F}{{\mathbb F}}

\newcommand{\Z}{{\mathbb Z}}

\newcommand{\id}[1]{\langle #1 \rangle}

\newcommand{\lm}{\lambda}

\DeclareMathOperator{\Char}{char}

\usepackage{color}

\numberwithin{equation}{section}
\numberwithin{table}{section}

\title{Cyclic and Constacyclic codes over $\Z_4+i\Z_4$}

\author{Miguel Martín}
\author{Ekin \"Ozman}

\date{}

\keywords{constacyclic codes; cyclic codes; quaternary codes; codes over finite chain rings.}

\begin{document}

\maketitle

\begin{abstract}
In this paper, we study cyclic and constacyclic codes over the finite chain ring $R=\Z_4 + i \Z_4$, where $i^2 = -1$. We prove that all constacyclic codes over $R$ are equivalent to cyclic codes. An algorithm to obtain generators for all simple root constacyclic codes over $R$ is presented. Using a Gray map we then obtain linear $\Z_4$ codes from constacyclic codes over $R$. We present new best linear $\Z_4$ codes found using this method.
\end{abstract}

\subsection*{Keywords:} constacyclic, cyclic and quaternary codes; codes over finite chain rings.

\section{Introduction}
Coding theory is an important area of study in modern mathematics, with a great range of applications from engineering to information theory.
Although most of the research was initially on codes over finite fields, the focus has evolved since to other kinds of alphabets, codes over different types of finite rings or over algebraic curves are of great relevance in modern mathematics \cite{book}. In this paper, we will be concerned with codes over a finite chain ring. In particular, an extension of $\Z_4$.
Codes over the ring $\Z_4$ have been particularly relevant in the study of coding theory. This stems from a series of discoveries around the 1990's, starting with a very relevant paper by Hammons et al. \cite{Hammons1994}, which related linear $\Z_4$ codes to some good non-linear binary codes. Moreover, codes over $\Z_4$ have numerous applications; in bio-informatics, for DNA sequencing, these codes have been studied \cite{Milenkovic2006}, \cite{Dimopoulou2020}. They are also relevant in cryptography and security as they give sequences of high linear complexity \cite{Helleseth1999}. For thorough literature on the basics of $\Z_4$ codes and their properties, the reader is invited to see the following book by Wan \cite{wan1997quaternary}.

This also motivated the research of codes over extensions of $\Z_4$, one of the most studied areas of codes over finite chain rings. This is not only for their properties as codes themselves, but because of the possibility to map them to $\Z_4$ codes. For instance, linear and cyclic codes over $\Z_4[u]/\langle u^2 \rangle$ have been studied by Yildiz, Karadeniz and Aydin in \cite{Yildiz2010} and \cite{Yildiz2014}. Other such extensions studied are $\Z_4[u]/\langle u^2-u \rangle$ by Bandi and Bhaintwal in \cite{Bandi2014}, and constacyclic codes over $\Z_4[u]/\langle u^2-2 \rangle$ by Aydin, Cengellenmis and Dertli in \cite{Aydin2024}.

In this paper, we consider simple root constacyclic codes over the ring $R=\Z_4[u]/\langle u^2+1 \rangle$ and study their properties. In Section \ref{prelims}, we introduce the relevant theory of finite chain rings and constacyclic codes, as well as studying and characterizing the ring $R$ in Section \ref{sec: ring R}. We prove that every constacyclic code over $R$ is equivalent to a cyclic code in Section \ref{equivcyclic}. An iterative algorithm to factor $x^n-\lm$ in $R[x]$ is presented in \ref{lifting} and used in Section \ref{sec: findings} to obtain new best $\Z_4$ codes. All the computations were done using Magma \cite{magma} and the relevant code can be found at 

\url{https://github.com/miguelmartin005/Cyclic-and-Constacyclic-Codes-Over-Z4-iZ4}{}.

\section{Preliminaries}\label{prelims}
A \emph{chain ring} is a ring that has its ideals totally ordered by inclusion. In the simple case where we consider a finite and commutative ring with unity, there are various equivalent characterizations.
\begin{proposition}[cf. \cite{Dinh}, Proposition 2.1]\label{charact}
    For a finite commutative ring $R$ with unity the following are equivalent:
    \begin{enumerate}
        \item $R$ is a chain ring.
        \item $R$ is a local principal ideal ring.
        \item $R$ is a local ring with principal maximal ideal $\langle \gamma \rangle$. In this case, $R$ has a total of $t+1$ ideals, all of the form $\langle \gamma ^i \rangle$ for $0 \leq i \leq t$, where $t$ is the nilpotency index of $\gamma $.
        \[\langle0 \rangle \subsetneq \langle \gamma ^{t-1} \rangle \subsetneq ... \subsetneq \langle \gamma ^2 \rangle \subsetneq \langle \gamma  \rangle \subsetneq R\]
    \end{enumerate}
\end{proposition}

\subsection{The chain ring $R =\Z_4+i\Z_4$}\label{sec: ring R}

Let $\mathbb Z_4$ denote the integers modulo $4$. The chain ring we will focus in this paper is $R =\Z_4+i\Z_4$ which is a finite commutative ring of characteristic four. It contains 16 elements.
The units of $R$ are precisely $r=a+bi \in R$ where $a+b \equiv 1 \mod 2$.
\[R^\times = \{1,i,3,3i,1+2i,2+i,3+2i,2+3i\}.\]
Any element which is not a unit will be inside the unique maximal ideal $\id{1+i}$. Then by Proposition \ref{charact}, we have that $R$ is a chain ring and all ideals of $R$ are generated by powers of $1+i$. We find the following three non-trivial ideals:
\begin{align*}
    \id{1+i} &= \{0,2,2i,1+i,2+2i,1+3i,3+i,3+3i\},\\
    \id{2} &= \{0,2,2i,2+2i\},\\
    \id{2+2i} &= \{0,2+2i\}.
\end{align*}
 Therefore, the local ring $R$ has a chain of inclusion of ideals as follows
\[ \langle0 \rangle \subsetneq \langle2+2i \rangle \subsetneq \langle 2 \rangle \subsetneq \langle 1+i \rangle \subsetneq R\]
Moreover, the residue field of the ring $R$ is $k =R/\id{1+i} \cong \F_2$.

\subsection{Basics of constacyclic codes}

Throughout this section, we consider $R$ to be a finite commutative chain ring with unity and we denote its residue field by $k$. 

A linear code over $R$ of length $n$ is an $R$-submodule of $R^n$, we say a linear code is \emph{free} when it is free as a submodule.

\begin{defn}
    Let $\lm$ be a unit in $R$, the map $\sigma_\lm : R^n \to R^n$ defined by
    \[\sigma_\lm(x_1,...,x_n) = (\lm x_n,x_1,...,x_{n-1})\]
    is called the \textbf{constacyclic shift by $\lm$}. A linear code $C \subset R^n$ is \textbf{$\lm$-constacyclic} when it is closed under $\sigma_\lm$, i.e. if for every $c \in C$ we have  $\sigma_\lm(c) \in C$. A constacyclic code over a finite chain ring $R$ is \textbf{simple root} if its length $n$ is coprime to the characteristic of $R$, $\Char(R)$ and \textbf{repeated root} otherwise.
\end{defn}

The notion of constacyclic codes is a generalization of the concept of cyclicity; in particular, cyclicity and $1$-constacyclicity are the same. Similarly to the cyclic case, we identify codewords with equivalence classes of polynomials via:
\[R^n \to R[x]/\id{x^n-\lm}\]
\[(c_0,...,c_{n-1}) \mapsto c_0+c_1x+...+c_{n-1}x^{n-1} + \id{x^n-\lm}.\]
When seeing codewords as classes of polynomials in $R[x]/\id{x^n-\lm}$, it follows that multiplication by $x$ corresponds to a $\lm$-constacyclic shift and thus, $\lm$-constacyclic codes are precisely ideals of the ring $R[x]/\id{x^n-\lm}$ (cf. \cite{book}, Proposition 17.3.1).

For some polynomial $f(x) \in R[x]$, we denote by $\overline f \in R[x]/\id{x^n-\lm}$ its class modulo $\id{x^n-\lm}$. Conversely, for any class $\overline f \in R[x]/\id{x^n-\lm}$, we write $f$ to be its unique minimal degree representative.

\begin{thm}\label{freee}
    For every monic polynomial $g(x) \mid x^n-\lm$ in $R[x]$, the code $\langle \overline g \rangle \subset R[x]/\langle x^n-\lm \rangle$ is a free $\lm$-constacyclic code of length $n$ and dimension $n-\deg(g)$ over $R$.
\end{thm}

\begin{proof}
    % The ideal $\langle g \rangle \lhd R[x]$ contains $x^n -a$ as $g$ is a divisor. This is thus in bijection with the ideal $\langle \overline{g} \rangle \lhd R[x]/\langle x^n -a \rangle$.

    To show that for any $g(x)$ dividing $x^n-\lm$, the code $\id{\overline g}$ is free of dimension $n-\deg(g)$, We first show the set $\{\overline g,x\overline g,...,x^{n-d-1}\overline g\}$ is a basis, where $d := \deg(g)$. Assume for $a_i \in R$ not all zero, we have
    \[a_0\overline g + a_1x\overline g + ...+a_{n-d-1}x^{n-d-1}\overline g = \overline 0,\] which means
    \[(x^n-\lm)\mid  a_0g +a_1xg +...+a_{n-d-1}x^{n-d-1}g.\]

    However, the highest degree appearing in the sum will be $n-1$, therefore this is only possible if the sum itself is zero. But then, the only term of degree $n-1$ will have coefficient $a_{n-d-1}$, implying $a_{n-d-1}=0$. Which then would mean the only term in the sum with coefficient $n-2$ would be $a_{n-d-2}$, giving this one is also zero. Inductively this would imply every $a_i = 0$, proving linear independence.

   To show that the set $\{\overline g,x\overline g,...,x^{n-d-1}\overline g\}$ spans the code $\langle \overline{g} \rangle$, notice that if $f + \langle x^n-\lm \rangle \in \langle \overline{g} \rangle$ and $\deg(f) <n$, then $f = gh$ for some $h \in R[x]$ of degree smaller than $ n-d$. Thus, to span any possible polynomial in this ideal, we need to multiply the generator $g$ by polynomials of degree up to $n-d-1$. In conclusion, since this set is linearly independent and spans any polynomial in the ideal, it is a basis of size $n-d$ for $\langle \overline g \rangle$.
\end{proof}

In fact, it turns out that any free simple root constacyclic code will be of this form (cf. \cite{Norton2000-it}, Proposition 4.11 and Remark 4.12). Therefore, if we understand the factorization of the polynomial $x^n-\lm \in R[x]$ for $n$ coprime to $\Char(R)$, we can generate all free simple root $\lm$-constacyclic codes over $R$. For this, we make use of Hensel's Lift.

\begin{theorem}[Hensel's Lift (cf. \cite{Dinh}, Lemma 2.4)]
    Let $f \in R[x]$ and $\tilde f\in k[x]$ its reduction to the residue field. Assume $\tilde f = g_1g_2...g_r$ where $g_i$ are pairwise coprime polynomials in $k[x]$. Then there exist pairwise coprime polynomials $f_i \in R[x]$ for $i \in \{1,...,r\}$ such that $f = f_1f_2...f_r$ and $\tilde f_i = g_i$ for all $i \in \{1,...,r\}$.
\end{theorem}

Usually we do not have a unique factorization over $R$. However, under certain conditions, we can conclude one, as described in the following lemma. Recall that a non-unit polynomial $f(x) \in R[X]$ is called \emph{basic irreducible} if
both $f$ and $\bar{f}$ are irreducible.

\begin{lemma}[cf. \cite{Norton2000-it}, Theorem 2.7 and \cite{Dinh}, Proposition 2.7]\label{uniquefact}
    If $f\in R[x]$ is monic and $\tilde{f}$ is square-free, then $f$ factors uniquely into monic, pairwise coprime, basic irreducible polynomials over $R[x]$.
\end{lemma}

Note that when these conditions are met, the factorization given by the Hensel's Lift will be precisely this unique factorization, as every $\tilde f_i = g_i$ so they are indeed monic, pairwise coprime, basic irreducible when $g_i$ are monic irreducible. This will always be the case when they are the factors of a monic $\tilde f$.

\begin{proposition}\label{uniquefact2}
    The polynomial $x^n-\lm$ factors uniquely into monic, pairwise coprime, basic irreducible polynomials in $R[x]$ when $\gcd(n,\Char(R))=1$.
\end{proposition}

\begin{proof}
    We need to check the assumptions in Lemma \ref{uniquefact}. Namely that $h = \widetilde{x^n - \lm}$ is square-free. Note $\lm$ is a unit so $\tilde\lm \neq 0$. Moreover, $h$ is separable as $h' = nx^{n-1}$ which is nonzero as $n$ is coprime to the characteristic. This means that $\gcd(h,h') =1$ as $h'$ is nonzero and only has $0$ as a root, which is not a root of $h$. Separability implies square-free, as being divisible by a polynomial twice will inevitably result in repeated roots.
\end{proof}

Proposition \ref{uniquefact2} is precisely the reason we study simple root constacyclic codes, the restraint on the length $n$ is trivial to ensure this unique factorization accessible from the Hensel's Lift, that proves very useful when studying the ideals of $R[x]/\langle x^n-\lm \rangle$. Also note that,  this motivates the choice of $\lm$ to be a unit, as it is necessary for the above proof.

\section{Equivalence of Constacyclic Codes Over $R$}\label{equivcyclic}

With the theory of free constacyclic codes over finite chain rings laid out, we turn our attention to the object of study in this paper. From now on we let $R$ be the ring defined as
\[ R:= \mathbb{Z}_4+ i\Z_4 \cong \mathbb{Z}_4[u] /\langle u^2+1 \rangle.\]

    \begin{defn}A monomial transformation between two $R$-linear codes of length $n$, $C_1$ and $C_2$ is a surjection of the form: 
    \[ T: C_1 \to C_2\]
    \[ (x_1,x_2,...,x_n) \mapsto (u_1x_{\sigma(1)}, u_2x_{\sigma(2)},..., u_nx_{\sigma(n)})\]
    Where $u_i$ are units and $\sigma \in S_n$ is a permutation \cite{eqcodesoverfiniterings}. 
    \end{defn}
    
    We say two codes $C_1,C_2 \subset R^n$ are \emph{equivalent} if there exists a monomial transformation $T: C_1 \to C_2$. It follows easily that when two codes $C_1$ and $C_2$ are equivalent, they have the same Hamming weight. Moreover, if either is free, then both are, and their dimension is equal.

In this section we show that every simple root $\lambda$-constacyclic code over $R$ is equivalent to a cyclic code over $R$. We will use the nice properties of the units of $R$ to do so. To make the work easier, denote

\[ U_1 = \{a+bi \in R^\times: a \equiv 1 \mod 2\} = \{1,3,1+2i,3+2i\},\]
\[U_2 = \{a+bi \in R^\times: a \equiv0\mod2\} = \{i,3i, 2+i, 2+3i\}.\]
Notice that $R^\times$ is the disjoint union of $U_1$ and $U_2.$

\begin{lemma}\label{squareunitsR}
    Let $\lm\in R^\times$. Then $\lm^2 = 1$ if $\lambda \in U_1$ and $\lm^2 = 3$ if $\lambda \in U_2$.
\end{lemma}
\begin{proof}
    Take $\lm = a+bi \in R^\times$. As precisely one of $a,b$ is even, $(a+bi)^2 = a^2 -b^2$. The square of every odd element in $\Z_4$ is $1$ and square of every even element is $0$. Thus $a^2-b^2$ will be $1$ if $a$ is odd and $3$ if $b$ is odd.
\end{proof}

\begin{corollary}\label{caseseq}
    Let $n$ be an odd positive integer. Then we have $\lm^n = \lm$ if and only if 
    \begin{itemize}
        \item $\lm \in U_1$ or
        \item $\lm \in U_2$ and $n \equiv 1 \mod 4$.
    \end{itemize}
         Similarly, we have  $\lm^n = -\lm$ if and only if $\lm \in U_2$ and $n \equiv 3 \mod 4$.
   
\end{corollary}

% \begin{proof}
%     By Lemma \ref{squareunitsR}, $\lm^m = 1$ for all positive even integers $m$ and units $\lm$, except when $\lm \in U_2$ and $m \equiv 2 \mod 4$, where $\lm^m =-1$. This implies for odd positive integers $n$, $\lm^n = \lm$ except in that case, where $\lm^n = -\lm$.
% \end{proof}
With this in hand, we can set up monomial transformations to cyclic codes.

\begin{proposition}\label{equivalencesR}
    Let $n$ be an odd positive integer. For any unit $\lm$ let $\mu = \pm \lm$ such that $\lm = \mu^n$ according to Corollary \ref{caseseq}. Then the following map is a ring isomorphism that is also a monomial transformation:
    \[\eta: R[x]/\langle x^n-\lm \rangle \to R[x]/\langle x^n- 1 \rangle\]
        \[\overline g(x) \mapsto \overline g(\mu x).\]

\end{proposition}

\begin{proof}
    First we check that $\eta$ is well-defined,

\[\eta(g(x) + h(x)(x^n-\lm) + \langle x^n - \lm \rangle) = g(\mu x) + h(\mu x)(\mu ^n x^{n}-\lm) + \langle x^n -1\rangle =\]
\[g(\mu x) + h(\mu x)(\lm x^n -\lm) + \langle x^n - 1 \rangle =
g(\mu x) + \langle x^n - 1 \rangle = \eta(g(x) + \langle x^n - \lm \rangle)\]

So the choice of representative does not matter.
This is trivially a homomorphism, to argue this function is an isomorphism one easily finds
\[\eta(g(x) + \id{x^n-\lm}) = \id{x^n-1} \leftrightarrow g(\mu x) \in \id{x^n-1} \leftrightarrow g(x) \in \id{x^n-\lm}\]
so the kernel is trivial, and surjectivity follows from the fact that their cardinality is equal.
The map $\eta$ is clearly a monomial transformation, as each term of the polynomial is mapped to a scalar multiple of itself.

    % To see that its a monomial transformation is not hard, each term $\alpha_k x^k$ of $g$ will just be sent to $\lm^{-k}\alpha_k x^k$, so this is a scaling by units of each coefficient.
    % We need to check well-definedness, but first notice its clear that this respects addition and multiplication. Moreover

    % \[x^n-1 + \langle x^n-1 \rangle \mapsto \lambda^{-n}x^n-1 + \langle x^n - \lm \rangle = \lm (x^n-\lambda) + \langle x^n-\lm \rangle = 0 + \langle x^n - \lm \rangle\]

    % With this we essentially get well-definedness as it implies adding $x^n-1$ to the representative we choose does not alter the result. So it is independent of representative.
    % It is a bijection as the map is invertible by $h + \langle x ^n - \lm \rangle \mapsto h(\lm x) + \langle x^n - 1 \rangle$.
\end{proof}

% The case $\lm^n = - \lm$ case is very similar, as seen in the following Theorem.

% \begin{theorem}
%     Let $\lm$ and $n$ as in the second case of Corollary \ref{caseseq}. Where $\lm^n = - \lm$ and $\lm^{-1} = - \lm$. Then the following map is a ring isomorphism that is a monomial transformation.

%     \[\eta: R/\langle x^n - 1 \rangle \to R/\langle x^n - \lm \rangle\]
%     \[g(x) \mapsto g(\lambda x)\]
% \end{theorem}
% \begin{proof}
%     We skip most of the proof as it follows in a completely analogous way, we just do the check for well-definedness.

%     \[x^n-1 + \langle x^n-1 \rangle \mapsto \lm^nx^n - 1 + \langle x^n - \lm \rangle = -\lm(x^n - \lm) + \langle x^n - \lm \rangle = 0 + \langle x^n - \lm \rangle\]
% \end{proof}

Having an isomorphism that is a monomial transformation implies that restricting this map (or its inverse) to any $\lm$-constacyclic code $C \subset R/\langle x^n - \lm \rangle$ we get an equivalent cyclic code and vice-versa. Therefore, the cyclic and $\lm$-constacyclic codes for any unit $\lm$ are essentially the same.

Thus, restricting to cyclic codes is sufficient for finding constacyclic codes with good parameters over 
$R$.

\section{Hensel's Lift: Factoring $x^n- \lm$ over $R[x]$}\label{lifting}

In this section, the known factorization of $x^n-1$ over the residue field $k\cong \F_2$ will be lifted to the unique factorization of any $x^n-\lm$ as described in Proposition \ref{uniquefact2}. As usual, we assume $\lm$ to be a unit of $R$ and $n$ coprime to the characteristic of $R$, this means we \textbf{consider $n$ to be odd throughout this section}.

First, we cover the case where $\lm = 1$, starting from the factorization of $x^n-1$ in $\F_2[x]$ and extending it to $R[x]$ with a simple iterative procedure. Then we map this factoring to the rest of cases for $\lm$ by exploiting the isomorphisms found in Proposition \ref{equivalencesR}. Finally, using this, we develop a Magma script to find this factorization for any such unit $\lm$ and length $n$.

    We have the following ring isomorphism
    \[R/\langle 2 \rangle \cong \F_2 + i \F_2\]
This intermediate ring between $\F_2$ and $R$ will be very useful for this process, as it turns out the factoring of $x^n-1$ over $R/\langle 2 \rangle$ is trivial. This means the only challenge is to fix the error modulo $2$ from this factoring. Note when we write $p(x)$ modulo $a \in R$, it is meant its reduction modulo the ideal $aR[x]$.
We will use the inclusion homomorphism to prove the triviality of the factoring over the intermediate ring $\F_2 + i\F_2$,
\[\Phi: \F_2[x] \to (\F_2 + i \F_2)[x]\]
this map keeps the polynomial the same, moving the coefficients to the bigger ring.

\begin{proposition}\label{skiplift}
    Let $f$ be a polynomial in $\F_2[x]$ and consider a factorization $f =\hat g\hat h$ for some coprime $\hat g, \hat h \in \F_2[x]$. If we take the polynomials $g = \Phi(\hat g)$ and $h = \Phi(\hat h)$, then $\Phi (f) = gh$ and $g,h$ are coprime.
\end{proposition}

\begin{proof}
    The homomorphism $\Phi$ preserves multiplication, therefore
    \[\Phi(f) = \Phi\left( \hat g\hat h\right) =  g h\]
    Moreover, as $\hat g$ and $\hat h$ are coprime, we have that there exist polynomials $\hat \alpha, \hat\beta \in \F_2[x]$ such that
    \[1 = \hat \alpha \hat g + \hat\beta\hat h\]
    Then define $\alpha = \Phi(\hat \alpha)$ and $\beta =\Phi(\hat \beta)$. Taking the homomorphism $\Phi$, we get
    \[1 = \Phi(1) = \Phi(\hat \alpha \hat g + \hat\beta\hat h) = \alpha g + \beta h\]
    Concluding $g,h$ are coprime as well.
\end{proof}

In particular, as mentioned before, if we take $f = x^n-1 \in \F_2[x]$, then we clearly find $\Phi(f) = x^n-1\in (\F_2+i\F_2)[x]$ so we already have a factorization of the polynomial $x^n-1$ over the ring $R/\langle 2 \rangle$ simply from the well known factorization over the field $\F_2$. All that is left is to lift this factorization to the ring $R$.
One important aspect about the proof of Proposition \ref{skiplift} is that the polynomials $\alpha,\beta$ in the Bezout's identity can be directly obtained from $\hat \alpha, \hat \beta$ for the same identity in $\F_2[x]$, this is used in the Magma script to produce new codes.

Now consider a monic reducible polynomial $f \in R[x]$ with real coefficients. If we take its reduction modulo $2$, we can obtain  some $g,h\in (\F_2+i\F_2)[x]$ by Proposition \ref{skiplift} such that
\[f \equiv gh \pmod { 2 }\]
\[\alpha g + \beta h \equiv 1 \pmod { 2 } \]
Where $g,h$ are coprime factors coming from the factoring of $f \pmod{1+i}$ and $\alpha, \beta$ as obtained in Proposition \ref{skiplift}. Note we require $f$ to have real coefficients so that its reduction $f \pmod 2$ is in the image of the canonical map $\Phi$, in particular $\Phi(f \pmod{1+i}) = f \pmod 2$, so we can apply the factoring in Proposition \ref{skiplift}. Then we can write
\[\alpha g + \beta h = 1+2c\]
\[f-gh =  \Delta = 2 \cdot k\]
Where, by slight abuse of notation, we now consider all polynomials $f,g,h$ to live inside $R[x]$ with the same coefficients. Here $c,k \in R[x]$ and $\Delta\in R[x]$ was introduced as the difference between the polynomial $f$ and the decomposition $gh$ in $R[x]$.
Now let $s = \alpha \cdot \Delta$, $t = \beta \cdot \Delta$.
\[gs + ht = \Delta + 2 \Delta c = \Delta + 4ck = \Delta\]
where $4ck = 0$ as $\Char(R) = 4$.
Moreover, as $s,t \equiv 0 \pmod 2$, we have $st = 0$ and therefore
\[(g +t)(h+s) = gh +(gs + ht) + st = f-\Delta + \Delta  = f\]

We found a valid factoring. However, there is one last change to make, as well as checking these are coprime. These choices of $s$ and $t$ might not always result in a monic decomposition. In general, $s$, $t$ contain high degree terms with coefficients of $2$, which prevent $g+t$ and $h+s$ from being monic. We would wish this resulting decomposition to be monic.

\begin{proposition}\label{s'prop}
  Let $s,t,g,h$ and $\Delta$ be as above.  Define the polynomials $s',t' \in R[x]$ as follows:  $s'$ is the remainder when dividing $s$ by $h$, which is possible as $h$ is monic, hence 
    \[s' \equiv s\mod h\]
    \[t' = \frac{\Delta-gs'}{h}\]
    Set $g_2 = g+t'$ and $h_2 = h+s'$, then we have $g_2$ and $h_2$ are monic and

    \[f = g_2h_2\]
\end{proposition}

\begin{proof}
    Let $s = s' +qh$, then
    \[t' = \frac{\Delta-gs'}{h} = \frac{gh + g(s'+qh) + ht-gs'}{h}=t+qg\]
    
    As $s'$ is the remainder after division by $h$, necessarily $\deg(s')<\deg(h)$. Then $h_2 = h+s'$ is the sum of a monic polynomial $h$ and a polynomial of lower degree, $s'$, thus making $h_2$ itself monic.
    Taking reductions modulo $2$ on $s = s' + qh$, we obtain $\overline{s'}=\overline{qh}$. Knowing $h$ and $s'$ are monic, $\deg(\overline h)=\deg(h)$ and $\deg(s') = \deg(\overline{s'})$. Then if $\overline q \neq \overline 0$, we would have $\deg(\overline{qh})=\deg(\overline{s'})> \deg(\overline{h})$. Implying $\deg(s')>\deg(h)$, a contradiction.
    Thus we can conclude $q \equiv 0 \pmod 2$. This fact implies both $s',t' \equiv 0 \pmod 2$ and expanding out we get
    \[g_2h_2 = (g+t')(h+s') = gh +ht'+gs' +t's'\]
    \[=gh +ht'+gs' =gh + h(t+qg)+g(s-qh) = gh + \Delta = f\]
    where we used $s't' = 0$ as we are working in characteristic $4$, and $\Delta =gs+ht$.
    Finally, as $f =g_2h_2$ and both $f$ and $h_2$ are monic, $g_2$ is monic as well.

\end{proof}

Now, recall the objective is to use this process to obtain a factorization of $x^n-1$ in the ring $R$. Of course, we will usually have more than two factors in the factorization of $x^n-1$ in $\F_2$, but if $x^n-1 = f_1...f_r$ in $\F_2$, one can take $g = f_1$, and $h=f_2...f_r$ and perform this process. Then we will obtain an irreducible $g_2$ and we can perform the process all over again with $h$, by taking the coprime polynomials $f_2$ and $f_3...f_r$ until we lift all individual polynomials $f_i$ one by one. This resulting factorization will be precisely the one ensured by the Hensel's Lift, as summarized in the next result.

\begin{lemma}\label{liftingthmyes}
    Let $x^n-1 = \tilde f_1...\tilde f_r$ in $\F_2[x]$, where $\tilde f_i \in \F_2[x]$ are irreducible and pairwise coprime, and let $f_i$ be the lifts obtained by the iterative process described in this section. Then the polynomials $f_i$ are monic, pairwise coprime, basic irreducible and $x^n-1 = f_1...f_r$ is the unique such factoring of $x^n-1$ in $R[x]$.
\end{lemma}

\begin{proof}
    It is enough to show that at every step $k$, the polynomials $g_2,h_2$, lifted from $g = \tilde f_k$ and $h = \tilde f_{k+1}...\tilde f_r$ are coprime, monic, and that $g_2 \equiv  g \pmod{1+i}$ and $h_2 \equiv h \pmod{1+i}$.
    Monicity was already proved in Proposition \ref{s'prop}, in that proof we also obtained that $s',t' \equiv 0 \pmod 2$. Because $\langle2 \rangle \subset \id{1+i}$, we then have
    \[g_2 = g+t' \equiv g\pmod{1+i}\]
    \[h_2 =h+s'\equiv h \pmod{1+i}\]
    To show coprimality of $g_2$ and $h_2$ take
    \[\alpha g_2+\beta h_2 = (\alpha g + \beta h) +(\alpha s' + \beta t') = 1+2V\]
    where $V$ is some polynomial in $R[x]$, this is because of the congruences $\alpha g + \beta h \equiv 1 \pmod2$ and $s',t' \equiv 0 \pmod 2$.
    But then as $(1+2V)(1-2V)=1$, we get
    \[\alpha(1-2V)g_2 + \beta(1-2V)h_2 = 1\]
    concluding $g_2$ and $h_2$ are coprime.
\end{proof}

We now have developed an algorithm to obtain the Hensel's Lift of the factorization of $x^n-1$ over $\F_2[x]$ to $R[x]$ for any odd positive integer $n$.

\begin{example}\label{lastrunning1}
    Take $n =15$, using this procedure with the Magma script, we can find the factoring of $x^{15}-1$ in $R[x]$ as
    \[x^{15}-1 = (x+3)(x^2+x+1)(x^4+2x^2+2x+1)(x^4+3x^3+2x^2+1)(x^4+x^3+x^2+x+1)\]
    One can verify this is a valid monic factoring, and Lemma \ref{liftingthmyes} confirms this is the unique factoring into monic, pairwise coprime, basic irreducible polynomials.
\end{example}

We would like to be able to extend this to factor $x^n-\lm$ for any unit $\lm$, this turns out to be a simple task, as factorization of $x^n-\lm$ is very closely related to the factorization of $x^n-1$ in this ring $R[x]$ as was seen in Proposition \ref{equivalencesR}.

\begin{lemma}\label{constamaplift}
    Define the map
    \[T_{\lm} : R[x] \to R[x]\]
    \[T_{\lm}(g)(x) = \lm^{-\deg(g)}g(\lm x)\]
    Then, if $x^n-1 = f_1...f_r$ is the Hensel's Lift of the factorization in $\F_2[x]$, $x^n-\lm = T_{\lm}(f_1)...T_{\lm}(f_r)$ is the Hensel's Lift of $x^n-\lm$.
\end{lemma}
\begin{proof}
    First, recall that for $\lm$ a unit, $\lm^{-1} = \pm\lm$ depending on the conditions in Lemma \ref{squareunitsR} and that, for all odd $n$, $\lm^{-n} = \lm$ as expressed in Corollary \ref{caseseq}. Then $T_{\lm}(x^n-1) = \lm^{-n}(\lm^nx^n-1) = x^n-\lm^{-n} = x^n - \lm$. This map is also multiplicative
    \[T_{\lm}(gh) = \lm^{-\deg(g)-\deg(h)}g(\lm x)h(\lm x) = T_{\lm}(g)T_{\lm}(h)\]
    Therefore, $x^n-\lm = T_{\lm}(f_1)...T_{\lm}(f_r)$. These polynomials are monic by construction, as by definition of $T_{\lm}$, if $g$ is monic then so is $T_{\lm}(g)$. Note this map can be inverted by
    \[T_{\lm}^{-1} : R[x] \to R[x]\]
    \[T_{\lm}^{-1}(g)(x) = \lm^{\deg(g)}g(\lm^{-1} x)\]
    which is also multiplicative in a similar way. Therefore, $g, h$ are coprime if and only if $T_{\lm}(g), T_{\lm}(h)$ are coprime, as B\'ezout's identity can be translated from one side to the other using the map $T_{\lm}$ or its inverse.
    
    Finally, they are also basic irreducible. To see this, recall reduction modulo $1+i$ sends units to $1$ and non units to $0$. Then
    \[\widetilde{T_{\lm}(f_i)} = \widetilde{\lm^{-\deg(f_i)}f_i(\lm x)} = \tilde{f_i}\]
    so $T_{\lm}(f_i)$ is basic irreducible if and only if $f_i$ is basic irreducible as their reduction to the residue field is the same.

    Concluding, if $x^n-1 = f_1...f_r$ is the factoring of $x^n-1$ into monic, pairwise coprime, basic irreducible polynomials; then $x^n-\lm = T_{\lm}(f_1)...T_{\lm}(f_r)$ is the factoring of $x^n-\lm$ into monic, pairwise coprime, basic irreducible polynomials.
\end{proof}

\begin{example}\label{lastrunning2}
    We found the factoring of $x^{15}-1$ in Example \ref{lastrunning1}. Using these maps we find
    \[x^{15}-i = T_i(x^{15}-1) = (x+i)(x^2+3ix+3)(x^4+2x^2+3ix+1)(x^4+ix^3+2x^2+1)(x^4+3ix^3+3x^2+ix+1)\]
    This factoring is again, the unique factoring into monic, pairwise coprime, basic irreducible polynomials. By choosing the proper value of $\lambda$ and using the Magma script one can find this factoring.
\end{example}

Therefore, for any odd positive integer $n$ and any unit $\lm$, we can obtain the desired factorization of $x^n-\lm$ in $R[x]$ by first factoring $x^n-1$ using the iterative process  in Lemma \ref{liftingthmyes}, and then mapping it to the factorization of $x^n-\lm$ using the map $T_{\lm}$ in Lemma \ref{constamaplift}.

\section{Finding $\Z_4$ codes}\label{sec: findings}

The aim of this last section is to use the procedure developed in the previous section to obtain all free constacyclic codes over $\Z_4 + i \Z_4$ and map them to $\Z_4$ codes. As mentioned before, one special aspect about quaternary codes is the ability to map them to binary codes, done via the Gray map \cite{wan1997quaternary}.
\[ \phi: \Z_4 \to \F_2^2\]
\[0 \mapsto 00\]
\[1 \mapsto 01\]
\[2 \mapsto 11\]
\[3 \mapsto 10\]
Let $\phi = (\phi_1, \phi_2)$ be the decomposition of $\phi$ on its two coordinates. Then we extend the Gray map to $\Z_4^n$ as
\[\phi :\Z_4^n \to \F_2^{2n}\]
\[(x_1,...,x_n) \mapsto (\phi_1(x_1),...,\phi_1(x_n),\phi_2(x_1),...,\phi_2(x_n))\]
This map is clearly a bijection, but it is not an additive group homomorphism, this is why linearity is often lost and some non-linear binary codes can be expressed as images of the Gray map.
While the Hamming weight is still well defined in $\Z_4$ codes, in the spirit of this Gray map, it is common to rather use the Lee weight.

\begin{defn}
    The \textbf{Lee weight} $w_L(c)$ of a codeword $c \in \Z_4^n$ is defined to be the quantity $w_L(c) := w_H(\phi(c))$. The \textbf{Lee distance} between two codewords $c_1,c_2 \in \Z_4^n$ is defined as $d_L(c_1,c_2) := w_L(c_1-c_2)$. The \textbf{minimum Lee distance} of a code $C \subset \Z_4^n$, denoted $d_L(C)$, is the smaller Lee distance that appears between codewords in $C$.
\end{defn}

Alternatively, from its definition it is clear that $w_L(c) = n_1(c) +2n_2(c) + n_3(c)$ where $n_i(c)$ is the number of appearances of the element $i$ in the codeword $c$.

Similarly, we will map codes over $R$ to codes over $\Z_4$ in an analogous way. These maps are also referred to as Gray maps. The following Gray map is commonly used to map from extensions of $\Z_4$ \cite{Aydin2024}
\[\varphi:R \to \Z_4^2\]
\[a+bi \mapsto (b,a+b)\]
extended component-wise as
\[\varphi:R^n \to \Z_4^{2n}\]
\[(x_1,...,x_n) \mapsto (b_1,...,b_n,a_1+b_1,...,a_n+b_n)\]
Although we will sometimes use the alternative extension map
\[\varphi':R^n \to \Z_4^{2n}\]
\[(x_1,...,x_n) \mapsto (b_1,a_1+b_1,...,b_n,a_n+b_n)\]
The extension $\varphi$ is commonly used, as mapping cyclic codes under it gives rise to quasi-cyclic codes (see Definition 4.10 and Theorem 4.12 in \cite{Aydin2024}). We introduce $\varphi'$ as we observed that, in some cases, using this alternative Gray map gives rise to cyclic $\Z_4$ codes. Evidently, for a code $C$ over $R$, $\varphi(C)$ and $\varphi'(C)$ are equivalent $\Z_4$ codes, as they simply differ by a permutation of coordinates.

\begin{lemma}\label{grayproperties}
The following properties hold for the Gray map $\varphi$.
\begin{itemize}
    \item The Gray map $\varphi$ is a $\Z_4$-module isomorphism.
    \item Suppose $C\subset R^n$ is a free $R$-linear code of dimension $k$ and basis $\{v_1,...,v_k\}$. Then $\varphi(C)\subset \Z_4^{2n}$ is a free $\Z_4$-linear code of dimension $2k$ and basis \\
    $\{\varphi(v_1),...,\varphi(v_k),\varphi(i\cdot v_1),...,\varphi(i \cdot v_k)\}$.
\end{itemize}
    
\end{lemma}

\begin{proof}
    \leavevmode
    \begin{itemize}
    
    \item It is enough to check the unique component map. It is routine to check the map is additive and $\Z_4$-linear.
    Moreover, $\varphi(a+bi) = (b,a+b)=(0,0)$ if and only if $a+bi=0$, so $\ker\varphi = \{0\}$. Concluding, as $\#R =16 = \#\Z_4^2$ the map $\varphi$ is a bijection.

    \item Notice for any free $R$-linear code $C\subset R^n$ of dimension $k$, its dimension as a $\Z_4$-module is $2k$. In particular, if  $\{v_1,...,v_k\}$ is an $R$-basis for $C$, $\{v_1,...,v_k,i\cdot v_1,...,i \cdot v_k\}$ is a $\Z_4$-basis for $C$.

    As $\varphi$ is a $\Z_4$-module isomorphism, it maps free submodules to free submodules and basis to basis respectively.
    \end{itemize}

\end{proof}

Notice the same proof follows in a completely analogous way for the alternative Gray map $\varphi'$ we defined.

In particular, Lemma \ref{grayproperties} gives the generators of the free $\Z_4$-linear code $\varphi(C)$ from the generators of any free $R$-linear code $C$. Having an algorithm to obtain all generators of free constacyclic codes over $R$ in \ref{lifting}, we can obtain numerous free linear $\Z_4$ codes.

We ran this procedure to obtain all constacyclic codes over $R$ for all units $\lambda$ and all odd $n \leq 31$. Then we mapped to linear $\Z_4$ codes with the described Gray maps and compared with the current linear quaternary codes in the database \cite{Aydin2023}, \cite{Z4database}. This way we obtained various new best codes, these new codes are displayed in Table \ref{findings}. All codes displayed have better Lee distance than those previously known for the same length and dimension and have been added to the database; except for codes $1,2$ and $7$ in Table \ref{findings}, which we kept in our findings as they are $\Z_4$-cyclic with tied parameters with the linear $\Z_4$ codes in the database.

\begin{table}[htbp]
    \centering
    \caption{New $\Z_4$ codes obtained}
    \label{findings}
    \renewcommand{\arraystretch}{1.2}
    \begin{adjustbox}{center}
    \hspace{-12mm}
    \begin{tabular}{|c|c|c|c|} % <--- The '|' symbols add the vertical lines
    \hline
    \textbf{Code} & \textbf{$\lambda$} & \textbf{Generator for the $R$ code} & \textbf{$\Z_4$ Parameters} \\ \hline
     $1$ & $1$ & $x^5 + 3x^4 + 2x^3 + x^2 + 2x + 3$ & [$30,20,4$]\\ \hline %CYCLIC
    $2$ & $i$ & $x^4 + 2x^2 + 3ix + 1$ &  [$30,22,4$] \\ \hline
    $3$ & $1$ & $x^2 + x + 1$ & [$30,26,2$] \\ \hline
    $4$ & $i$ & $x^8 +2x^6 +ix^5 +x^4 +3ix^3 +2x^2 +1$  & [$34,18,8$] \\ \hline
    $5$ & $i$ & $x^{10} + 2x^8 + 3ix^7 + 3x^6 + x^4 + 3x^2 + 2ix + 3$  & [$42,22,9$]\\ \hline
    $6$ & $i$ & $x^9 + ix^8 + x^7 + ix^6 + 2x^5 + 3ix^4 + 3x^3 + 3ix^2 + 2x + 3i$ & [$42,24,7$] \\ \hline
    $7$ & $1$ & $x^8 + 2x^7 + 3x^6 + 3x^5 + 3x^4 + 3x^3 + 3x^2 + 2x + 1$ & [$42,26,4$] \\ \hline %CYCLIC
    $8$ & $1$ & $x^2 + x + 1$ & [$42,38,2$] \\ \hline
    $9$ & $i$ & $x^{11} + 2ix^{10} + x^9 + 3x^7 + ix^6 + x^5 + 2ix^4 + 3x + 3i$ &  [$46,24,11$] \\ \hline
    $10$ & $1$ & $x^2+x+1$ & [$54,50,2$] \\ \hline
    $11$ & $i$ & \makecell[c]{$x^{20} + ix^{19} + 2x^{16} + 3ix^{15} + 3x^{14} + 3ix^{13} + 2ix^{11}+ x^{10}$ \\
    $+ 3ix^9 + 2x^8 + 3ix^7 + x^6 + 2ix^5 + 2x^4 + 2ix^3 + 3x^2 + 2ix + 1$}  & [$62,22,21$] \\ \hline
    $12$ & $i$ & $x^{15} + 3ix^{14} + 3x^{13} + 2ix^{12} + x^9 + ix^8 + 2x^7 + 2ix^6 + 2ix^4 + x^3 + 3i$ &  [$62,32,14$] \\ \hline
    $13$ & $1$ & $x^{11} + x^{10} + x^9 + 2x^7 + 3x^6 + 3x^5 + 2x^3 + 3$ & [$62,40,8$] \\ \hline
    $14$ & $i$ & $x^{10} + 3ix^9 + x^8 + x^6 + 3ix^5 + 2x^4 + 3ix^3 + 3x^2 + 3ix + 3$ &  [$62,42,8$] \\ \hline
    $15$ & $i$ & $x+3$ & \makecell[c]{[$n,n-2,2$] \\
    even $18 \leq n \leq 62$, \\ $n \neq 30$} \\ \hline
    \end{tabular}
    \end{adjustbox}
    
\end{table}

Of all the codes listed, various have been found to be $\Z_4$-cyclic under the Gray map $\varphi'$, these are listed with their $\Z_4$ generators in Table \ref{cyclicgenstable}. The rest have not been found to be cyclic under either map.

\begin{table}[htbp]
    \centering
    \caption{New $\Z_4$-cyclic codes generators}
    \label{cyclicgenstable}
    \renewcommand{\arraystretch}{1.3}
    \makebox[\textwidth][c]{
    \begin{tabular}{|c|c|c|} % <--- The '|' symbols add the vertical lines
    \hline
    \textbf{Code} & \textbf{Generator as a $\Z_4$-cyclic code} & \textbf{$\Z_4$ Parameters} \\ \hline
    $1$ & $x^{10} + 2x^8 + 3x^6 + 2x^4 +x^2+3$ & [$30,20,4$] \\ \hline
    $7$ & $x^{16}+2x^{14}+3x^{12}+3x^{10}+3x^8+3x^6+3x^4+2x^2+1$ & [$42,26,4$] \\ \hline
    $13$ & $x^{22} + 2x^{16} + x^{12}+x^{10} +2x^8 + 3x^4 +3x^2 +3$ & [$62,40,8$] \\ \hline
    $3,8,10$ & $x^4 + x^2 + 1$ & \makecell[c]{[$n,n-4,2$]\\
    $n \in \{30,42,54\}$} \\ \hline
    $15$ & $x^2+3$ & \makecell[c]{[$n,n-2,2$] \\
    even $18 \leq n \leq 62$, $n \neq 30$} \\ \hline
    \end{tabular}}
\end{table}

\newpage

\begin{example}
    To understand how these codes are obtained, we explain the second code in Table \ref{findings}.
    In Example \ref{lastrunning2}, we obtained the factoring of $x^{15}-i$. Thus, we can find free $i$-constacyclic $R$ codes of length $15$.
    In particular, the polynomial $g = x^4+2x^2+3ix+1$ generates an $i$-constacyclic code 
    \[C_g = \langle \overline{g} \rangle = \id{x^4+2x^2+3ix+1 + \id{x^n-i}}\]
    of dimension $k=15-4=11$.

    Mapping this code under the map $\varphi$ gives a free $\Z_4$-linear code of double length $2n =30$ and double dimension $2k = 22$ by Lemma \ref{grayproperties}. Finally, Magma finds $d_L(\varphi(C_g)) = 4$. 
    % This code has higher minimum Lee distance than the previously known quaternary codes of same length and dimension.
\end{example}

% \begin{example}
%     To understand how these codes are obtained, we explain the first code in the table.

%     Running the code in Appendix \ref{code2} we get can get all divisors of $x^{15}-1$ over $R$. One of them being $g=x^5+3x^4+2x^3+x^2+2x+3$, thus the $R$-cyclic code

%     \[C_g = \langle \overline g \rangle\]
% has $\dim(C_g) = 15-5 =10$ by Lemma \ref{dimensioncodering}. Now we map this code to $\Z_4$ with the map $\varphi'$, this code has distance $30$ and dimension $20$ by Theorem \ref{grayproperties}. Finally, running a command in Magma, we find its Lee distance to be $d_L = 4$. Thus, the resulting $\Z_4$ code has parameters $[30,20,4]$.
% \end{example}

For entry number 15 in Table \ref{findings}, these [$n,n-2,2$] codes exist for all even $n\geq 4$ with the same generator, and were found to be better than those in the database for $n = 18,22,26, 34, 38, 42, 46,50, 54, 58, 62$. In fact, these and some other cyclic codes found have an easily defined structure and exist for infinitely many $n$, as we explain next in Proposition \ref{cyclicexplained}.

\begin{proposition}\label{cyclicexplained}
    The following cyclic $\Z_4$ codes exist:

    \begin{enumerate}
    \item For any length $n$, the polynomial $x+3$ generates a cyclic $\Z_4$ code with parameters $[n,n-1,2]$.

    \item For any even length $n$, the polynomial $x^2+3$ generates a cyclic $\Z_4$ code with parameters $[n,n-2,2]$.

    \item For any length $n \equiv 0 \pmod 6$, the polynomial $x^4+x^2+1$ generates a cyclic $\Z_4$ code with parameters $[n,n-4,2]$.
    \end{enumerate}
\end{proposition}

\begin{proof}
    We will make use of Theorem \ref{freee}, which also holds for $n$ not coprime to the characteristic of the ring.

    \begin{enumerate}
        \item We will always have $1$ is a root of $x^n-1$. Thus, $x+3\mid x^n-1$ so we can generate a cyclic code with it of any length $n$ and dimension $n-1$.

        Now, we have $w_L(\overline{x+3}) = 2$. As the only cyclic code of minimal Lee distance $1$ is the trivial $[n,n,1]$ code, we must have $d_L =2$.

        \item For even $n$, both $1$ and $3$ are roots of $x^n-1$. As we have $x^2+3 = (x+1)(x+3)$, we have $x^2+3\mid x^n-1$. Thus this generates a cyclic code with dimension $n-2$ and the same exact argument for the minimal distance shows $d_L = 2$.

        \item First notice $x^6-1 =(x^2-1)(x^4+x^2+1)$. Now, if $n=6m$ we have $x^6-1\mid (x^6)^m-1$, so the polynomial $g =x^4+x^2+1$ always divides $x^n-1$ and thus generates a cyclic code of dimension $n-4$. Now note the codeword
        \[\overline g + 3x^2\overline g = \overline{3x^6+1} \in C_g\]
        has Lee weight $2$. As before, this implies $d_L = 2$.
    \end{enumerate}
\end{proof}

It is easy to check, by mapping the generators as in Lemma \ref{grayproperties}, that the codes of type $2$ as in Proposition \ref{cyclicexplained} can be obtained as in entry 15 of Table \ref{findings}, by mapping of the $R$-cyclic codes generated by $x+3$ using $\varphi'$.
Similarly, codes of type $3$ in Proposition \ref{cyclicexplained} can be obtained by mapping the $R$-cyclic codes of length divisible by $3$ generated by $x^2+x+1$ using $\varphi'$, as codes $3$, $8$ and $10$ in Table \ref{findings}.
Of course, due to an $[n,n-1,2]$ code existing for every $n$, these cyclic codes found of types $2$ and $3$, having Lee distance $2$ are not very relevant for most purposes. Nevertheless, the ones listed in Table \ref{findings} are better than those in the database, and therefore have been added.

Under the original Gray map $\phi$, one can map these $\Z_4$ codes to binary codes, often resulting in non-linear binary codes. This way a free $\Z_4$ code with parameters $[n,k,d_L]$ will be mapped to a binary code of length $2n$, minimal Hamming distance $d_L$ and size $4^k$.

\bibliographystyle{alpha} 
\bibliography{ref}

\newpage

\end{document}